\documentclass[12pt]{article}
\usepackage{authblk}
\usepackage{amsmath, amsthm, amssymb}
\usepackage{tikz}
\usepackage[caption=false,font=footnotesize]{subfig}

\newcommand{\xN}{x^{N}}
\newcommand{\XN}{X^{N}}
\newcommand{\yN}{y^{N}}
\newcommand{\YN}{Y^{N}}
\newcommand{\calX}{\mathcal{X}}
\newcommand{\calY}{\mathcal{Y}}
\newcommand{\zN}{z^{N}}
\newcommand{\ZN}{Z^{N}}
\newcommand{\calZ}{\mathcal{Z}}
\newcommand{\define}{\stackrel{\Delta}{=}}

\newcommand{\Pe}{P_{\mathcal{E}}}
\newtheorem{theorem}{Theorem}
\allowdisplaybreaks
\begin{document}
\baselineskip 12pt

\begin{center}
\textbf{\Large Information Theoretical Analysis of Identification based on Active Content Fingerprinting} \\
\vspace{5mm}
\begin{tabular}{ccc}  Farzad~Farhadzadeh$^1$  & Frans~M.J.~Willems$^1$ & Sviatoslav~Voloshinovskiy$^2$ \\[1.2\baselineskip]
\multicolumn{2}{c}{$^1$ Eindhoven University of Technology} & $^2$ University of Geneva\\
\multicolumn{2}{c}{Electrical Enginering Dept.} & Computer Science
Dept.\\
\multicolumn{2}{c}{Eindhoven, The Netherlands} & Geneva, Switzerland\\[1.2\baselineskip]
 \verb+f.farhadzadeh@tue.nl+ & \verb+f.m.j.willems@tue.nl+ &
 \verb+svolos@unige.ch+
\end{tabular} \end{center}

\begin{abstract}
  Content fingerprinting and digital watermarking are techniques that are used for content protection and distribution monitoring. 
  Over the past few years, both techniques have been well studied and their shortcomings understood. 
  Recently, a new content fingerprinting scheme called {\em active content fingerprinting} was introduced to overcome these shortcomings. 
  Active content fingerprinting aims to modify a content to extract robuster fingerprints than the conventional content fingerprinting. 
  Moreover, contrary to digital watermarking, active content fingerprinting does not embed
  any message independent of contents thus does not face host interference. 
  
  The main goal of this paper is to analyze fundamental limits of active content fingerprinting in an information theoretical framework. 
\end{abstract}

\section{Introduction}
Generally, identification systems \cite{Willet:2003} are facing numerous requirements related to identification rate, complexity, privacy, security as well as memory storage. 
To address the trade-off between these conflicting requirements, \emph{content fingerprints} are used \cite{Fridrich:1999}, \cite{Kalker:2002}. 
A content fingerprint is a short, robust and distinctive content description. 

In conventional content fingerprinting, the fingerprint is extracted directly from an original content and does not require any content modification that preserves the original content quality. 
In this sense, it can be considered as a {\em passive content fingerprinting} (PCFP). 
The extracted fingerprints resemble random codes, for which no efficient decoding algorithm is known as for structured codes. 
Moreover, the performance of PCFP in terms of identification rate is not satisfactory due to acquisition device imperfections.

For these reasons, \emph{active content fingerprinting} (ACFP) was proposed in \cite{Slava:WIFS:2012}, \cite{Farzad:ICASSP:2013}, where the basic idea was introduced and a feasibility study revealed higher performance w.r.t. PCFP and digital watermarking. 
ACFP by modifying digital contents takes the best from two worlds of content fingerprinting and digital watermarking to overcome some of fundamental restrictions of these techniques in terms of performance and complexity. 
In the proposed fingerprinting scheme, contents are modified in a way to increase the identification rate and reduce the decoding complexity with respect to conventional content fingerprinting.  

The main goal of this sequel is to analyze ACFP in an information theoretical
framework to investigate its fundamental limits in 
identification systems.  
In this paper, we investigate the identification capacity based on ACFP in which a
modified content can be modelled as an output of a discrete memoryless channel 
with an original content as its input. 
Moreover, we investigate the optimal encoding scheme under the assumptions that content sequences can be modeled as a Gaussian memoryless source and the observation channel as an additive white Gaussian. 
And, we introduce an optimal scheme that can achieve the identification capacity based on ACFP.

The rest of this paper is organized as follows.
Section~\ref{sec:model_desc_gen} presents the identification system based on active content fingerprinting  and we will state our main result.
Section~\ref{sec:proof} contains the proof of this result.
Concluding remarks follow in Section \ref{sec:conclusions}.

\textbf{Notations:} 
Throughout this paper, we adopt the convention that a scalar random variable is denoted by a capital letter $X$, a specific value it may take is denoted by the lower case letter $x$, and its alphabet is designated by the script letter $\mathcal{X}$. 
As for vectors, a capital letter $\XN$ with a corresponding superscript will denote an $N$-dimensional random vector $\XN = (X_1, \ldots, X_N)$. 
A lower case letter $\xN$ will represent its particular realization $\xN= (x_1, \ldots, x_N)$.
The expectation operator is designated by $E[\cdot]$. 

\section{Model Description}
\label{sec:model_desc_gen}
\begin{figure}[h]
\centering
\begin{tikzpicture}
  \begin{scope}
   \draw (-7,3) rectangle node{$Q_s(x)$} (-5.75,4) ;
  \draw[-latex] (-5.75,3.5) -- node[above]{$\XN(1)$} (-4.,3.5);
  \draw[] (-4,3) rectangle node{$\YN=e(\XN)$} (-1.5,4);
  \draw[-latex] (-1.5,3.5) -- node[above]{$\YN(1)$} (0.25,3.5);
  \end{scope}
  \begin{scope}[yshift=-1.5cm]
   \draw (-7,3) rectangle node{$Q_s(x)$} (-5.75,4) ;
  \draw[-latex] (-5.75,3.5) -- node[above]{$\XN(2)$} (-4.,3.5);
  \draw[] (-4,3) rectangle node{$\YN=e(\XN)$} (-1.5,4);
  \draw[-latex] (-1.5,3.5) -- node[above]{$\YN(2)$} (0.25,3.5);
  \end{scope}
  \begin{scope}[yshift=-3.5cm]
  \draw (-7,3) rectangle node{$Q_s(x)$} (-5.75,4) ;
  \draw[-latex] (-5.75,3.5) -- node[above]{$\XN(M)$} (-4.,3.5);
  \draw[] (-4,3) rectangle node{$\YN=e(\XN)$} (-1.5,4);
  \draw[-latex] (-1.5,3.5) -- node[above]{$\YN(M)$} (0.25,3.5);
  \end{scope}
  \node at (-6.375,1) {$\vdots$};
  \node at (-2.75,1) {$\vdots$};
  \begin{scope}[yshift=0.75cm]
  \draw[-latex] (-7,-3.5) -- node[above]{$W$} (-5.75,-3.5);
  \draw (-5.75,-4) rectangle node{$\YN=s(W)$} (-3.5,-3);
  \draw[-latex] (-3.5, -3.5) -- node[above]{$\YN(W)$} (-1.75,-3.5);
  \draw (-1.75,-4) rectangle node[]{$Q_c(z \mid y)$} (0,-3);
  \draw[-latex] (0,-3.5) -- node[above]{$\ZN$} (1.25,-3.5);
  \draw (1.25,-4) rectangle node{$\widehat{W}=d(\ZN)$} (3.5,-3);
  \draw[-latex] (3.5,-3.5) -- node[above]{$\widehat{W}$} (4.75,-3.5);
  \draw[-latex] (-5.5, -2.25) node[right]{$\YN(1)$} -- (-5.5, -3);
  \draw[-latex] (-4, -2.25) node[right]{$\YN(M)$} -- (-4, -3);
  \node at (-4.75, -2.75) {$\cdots$};
  \draw[-latex] (1.5, -2.25) node[right]{$\YN(1)$} -- (1.5, -3);
  \draw[-latex] (3, -2.25) node[right]{$\YN(M)$} -- (3, -3);
  \node at (2.25, -2.75) {$\cdots$};
  \end{scope}
\end{tikzpicture}
\caption{Model of an identification system using modified content-sequences.}
\label{fig:system}
\end{figure}
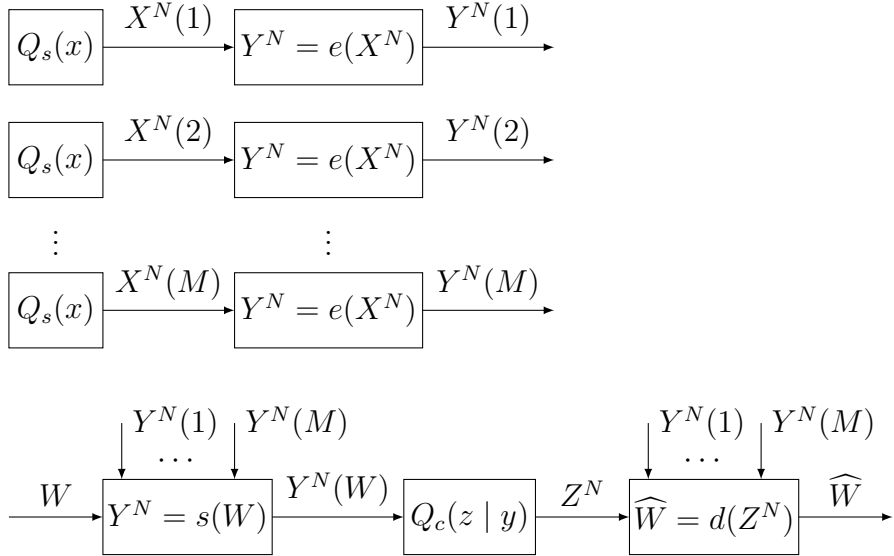

In an identification system, see Fig.~\ref{fig:system}, there are $M$ items indexed $w\in \{1,2, \cdots, M\}$ to be identified.
A randomly generated content-sequence (vector) of length $N$ corresponds to each item. 
This sequence has symbols $x_n,n=1,2,\cdots,N$ taking values in the discrete alphabet $\calX$, and the probability that content-sequence $\xN=(x_1,x_2,\cdots,x_N)$ occurs for item $w$ is
\begin{equation}
  \label{eq:src_dist}
\Pr\{\XN(w)=\xN\} = \prod_{n=1}^N Q_s(x_n),
\end{equation}
hence the components $X_1,X_2, \cdots, X_N$ are independent and identically distributed according to $\{Q_s(x), x\in\calX\}$. 
Note that this probability does not depend on the index $w$. 

An encoder $e(\cdot)$ transforms each content-sequence $\xN$ into a modified content-sequence $\yN=(y_1,y_2,\cdots,y_N)$, where $y_n,n=1,2,\cdots,N$ taking values in the discrete alphabet $\calY$.  
The distortion between modified content sequence and content-sequence cannot be too large.
The modification distortion $\overline{D_{xy}}$ is defined as 
\begin{equation}
  \overline{D_{xy}} = \frac{1}{N} E \left[ \sum_{n=1}^{N} D_{xy} (X_n , Y_n) \right],
  \label{eq:dist_def}
\end{equation} 
where $\{D_{xy} (x, y), x \in \calX , y \in \calY \}$ is the distortion matrix specifying the distortion per component. 
We assume that the distortion matrix has only finite non-negative entries. 
Moreover, we assume that all modified content-sequences are generated prior to the identification procedure. 
These modified sequences form a codebook that we call the ``database'' here. 
This database $C$ consists of the list of entries, hence
\begin{equation}
C =   \left( \yN(1),\yN(2),\cdots,\yN(M) \right).
\end{equation}

In the identification process the probabilities for the items to be presented for identification are all equal, hence
\begin{equation}
\Pr\{W=w\} = 1/M \mbox{ for } w\in \{1,2,\cdots,M\}.
\end{equation}
When item $w$ is presented for identification, its corresponding modified content-sequence $\yN(w)$ is ``selected'' from the database $C$ and presented to the system, hence
\begin{equation}
\yN = s(w).
\end{equation}
The system observes $\yN$ via a memoryless observation channel $\{Q_c(z|y), y\in\calY,z\in\calZ\}$, with discrete alphabet $\calZ$, and the resulting channel output sequence is $\zN=(z_1,z_2,\cdots,z_N)$, where $z_n \in \calZ$ for $n=1,2,\cdots,N$.
Now
\begin{equation}
\Pr\{\ZN=\zN|\YN(w)=\yN \} = \Pi_{n=1}^N Q_c(z_n|y_n).
\end{equation}
After observing $\ZN$, the decoder decides that $\ZN$ is related to which
modified content-sequence. 
If this is $\YN(w)$, the decoder outputs $\widehat{W}=w$.
The reliability of our identification system is measured by the error probability 
\begin{equation}
  \Pe = \Pr \{\widehat{W} \neq W\}.
  \label{eq:er_prob}
\end{equation}
\subsection{Statement of Result}
An identification rate-distortion pair $(R,\Delta)$ is said to be achievable if for all $\epsilon>0$ there exist for all $N$ large enough an encoder and a decoder such that
\begin{align}
  \label{eq:ach_R_D}
  \overline{D_{xy}} &\leq \Delta + \epsilon, \nonumber \\
\log_2(M) &\geq N(R-\epsilon), \mbox{ and} \nonumber \\
  \Pr\{ \widehat{W} &\neq W\} \leq \epsilon.
\end{align}

We are now ready to state the main result of this submission, the proof follows in section III.
\begin{theorem}
  \label{thm:main}
  The region of achievable rate-distortion pair $(R,\Delta)$ for the identification system based on modifies content-sequence is given by
\begin{align}
\bigg\{ (R, \Delta) :\, & R \leq I(Y;Z), \nonumber \\
      &\Delta \geq \sum_{x,y} Q_s(x)P_t(y\mid x)D_{xy}(x,y), \nonumber \\
      &  \mbox{for } P(x,y,z) = Q_s(x)P_t(y \mid x) Q_c(z \mid y)\bigg\}.
\end{align}
\end{theorem}
The ``capacity'' of identification based on ACFP, the maximum of possible identification rate, for a given distortion $\Delta$ is
given by
\begin{equation}
  C_{\text{ACFP}}(\Delta) = \max_{P_t(y \mid x):\sum_{x,y} Q_s(x)P_t(y \mid x) D_{xy}(x,y) \leq \Delta} I(Y;Z).
\end{equation}

\section{Proof}
\label{sec:proof}
The proof consists of the achievability part and a converse.
We start with the converse.
\subsubsection{Converse Part}
First, we define the random variable $I$ that takes values in $\{1,2,\cdots,N\}$
with probability $1/N$. 
Then the random triple $(X,Y,Z)$ is defined as $(X,Y,Z)\define (X_i,Y_i,Z_i)$ if
$I=i$. 
Hence, the joint distribution of $(X,Y,Z)$ is given by 
\begin{align}
  P(x , y , z) &= \frac{1}{N} \sum_{i=1}^N \Pr\{X_i=x , Y_i=y , Z_i = z\} \nonumber \\
  &= \frac{1}{N} \sum_{i=1}^N Q_s(x_i)P_t(y_i \mid x_i)Q_c(z_i \mid y_i)
  \nonumber \\
    &= Q_s(x)P_t(y \mid x)Q_c(z \mid y)
  \label{eq:jointDist}
\end{align}

Consider the $M$ number of modified contents. 
Using Fano's inequality $H(\widehat{W}\mid W) \leq F$ where $F = 1 + \Pr\{\widehat{W} \neq W\} \log_2(M)$, we have the following series of (in)equalities:
\begin{align}
  \label{eq:bndM}
  log_2 (M) &= H(W) \nonumber \\ 
  &\leq I (W ; \ZN , \YN (1), \cdots , \YN (M)) + F \nonumber \\
  &= I(W;\YN (1), \cdots , \YN (M)) + I(W;\ZN \mid \YN (1), \cdots , \YN (M)) + F \nonumber \\
  &\stackrel{(a)}{\leq} H(\ZN) - H(\ZN \mid \YN (1), \cdots , \YN (M) , W ) + F \nonumber \\
  &= H(\ZN) - H(\ZN \mid \YN(W))  + F \nonumber \\
  &\leq \sum_{i=1}^N H(Z_i) - \sum_{i=1}^N H(Z_i \mid Y_i(W)) + F \nonumber \\
  &= \sum_{i=1}^N I(Y_i(W); Z_i) + F \nonumber \\
  &= NH(Z \mid I) - NH(Z \mid Y, I) + F \nonumber \\
  &\stackrel{(b)}{\leq} N I(Y;Z) + F,
\end{align}
where (a) and (b) follow from the facts that conditioning reduces entropy and
the channel is memoryless. 

Now for the distortion part we have
\begin{align}
  \overline{D_{xy}} &= \frac{1}{N} E\left[\sum_{n=1}^N D_{x,y}(X_n,Y_n)\right] \nonumber \\
  &=\frac{1}{N} \sum_{n=1}^N \sum_{x_n,y_n} Q_s(x_n)P_t(y_n\mid x_n)D_{x,y}(x_n,y_n) \nonumber \\
  &=\sum_{x,y} Q_s(x)P_t(y\mid x)D_{x,y}(x,y) \nonumber \\
  &= D_{xy}(X,Y).
  \label{eq:dist}
\end{align}

Now, assume that $(R,\Delta)$ is achievable. Then $F \leq 1 + \epsilon \log_2(M)$ and $\Delta \geq \overline{D_{xy}} - \epsilon$. 
For all blocklengths $N$ and small enough $\epsilon > 0$, we obtain from
\eqref{eq:bndM} that 
\begin{equation}
  \label{eq:bndC}
  N(R-\epsilon) \leq \log_2 (M) \leq \frac{1}{1-\epsilon} (N I(Y;Z)+1),
\end{equation} 
for some $P(x,y,z) = Q_s(x)P_t(y \mid x)Q_c(z \mid y)$. 
If we now let $\epsilon \downarrow 0$ and $N \rightarrow \infty$, then we obtain the
converse of Theorem \ref{thm:main} from \eqref{eq:dist} and \eqref{eq:bndC}.

\subsubsection{Achievability}
We can only give an outline of the achievability proof here. 
For each content-sequence $\XN(w)$, a modified content-sequence 
$\YN(w)$ is generated using conditional distribution $P_t(y \mid x)$. 
These modified sequences are codewords in a random codebook that guarantee a rate that can be as large as
$I(Y;Z)$. 
The distortion is as expected $(i.e. D_{xy}(X,Y) )$ because of the law of large
numbers.
\subsection{Gaussian Source}
\label{sec:Gauss_src_gen}
Let's assume the content-sequences are distributed i.i.d. according to a Gaussian distribution with variance $V_X$ and mean zero. 
Moreover, the observation channel $Q_c(z \mid y)$ can be modelled as an additive
white Gaussian noise (AWGN) with variance $V_N$. 
\begin{theorem}
Considering distortion as the mean-squared error, the capacity of identification
based on ACFP is given by
\begin{equation}
  C_{\text{ACFP}}(\Delta) = \frac12 \log_2 \left( 1+ \frac{(\sqrt{V_X} + \sqrt{\Delta})^2}{V_N} \right),
  \label{eq:gauss_cap}
\end{equation}
that can be achieved by scaling the content by a factor $f$, i.e. $\YN=f \XN$,
such that $(f-1)^2V_X=\Delta$.
\end{theorem}
\begin{proof}
  We can upper bound the identification rate as follows
  \begin{align}
    I(Y;Z) &= h(Z) - h(Z \mid Y) \nonumber \\
    & = h(Z) - \frac12 \log_2 2\pi e V_N \nonumber \\
    & \leq \frac12 \log_2 2\pi e (V_Y+V_N) - \frac12 \log_2 2\pi e V_N \nonumber \\
    & \stackrel{(a)}{\leq} \frac12 \log_2 \left( 1+\frac{(\sqrt{\Delta} + \sqrt{V_X})^2}{V_N} \right),
    \label{eq:gauss_up}
  \end{align}
  where $V_Y = E[Y^2]$ and (a) follows from the fact that 
  \begin{equation*}
    E[(X-Y)^2] = V_Y+V_X-2E[XY] \leq \Delta.
  \end{equation*}
  $V_Y$ attains the maximum if the equality holds in the above equation and $Y$ is aligned
  with $X$, i.e. $\YN=f \XN$, such that $(f-1)^2V_X=\Delta$. 
  Note that by setting $\YN=f \XN$ equalities in \eqref{eq:gauss_up} hold. 
\end{proof}

Figure~\ref{fig:Rate_dist_gen} shows the capacity of identification systems using PCFP versus ACFP for different values of distortion $\Delta$. 

\begin{figure}[t]
  \centering
  \includegraphics[width=0.75\textwidth]{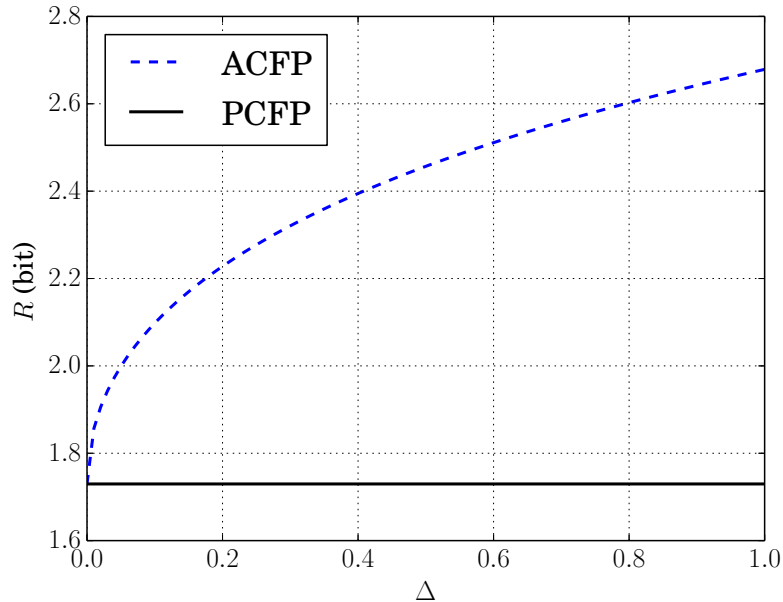}
  \caption{(a) the capacity of identification based on ACFP versus PCFP for a Gaussian source and AWGN with SNR=$10$ dB.}
\label{fig:Rate_dist_gen}  
\end{figure}

\section{Conclusions}
\label{sec:conclusions}
In this paper, we evaluated the capacity of identification systems based on active
content fingerprinting. 
In active content fingerprinting, the main goal is to modify a digital content to improve the performance in terms of identification rate. 
We assumed that the modification can be modeled by a memoryless channel. 
Then, we investigated the optimal encoding scheme under Gaussian assumptions that can achieve the identification capacity based on ACFP.

\bibliographystyle{IEEEtran}
\bibliography{acfp_it} 

\end{document}